\documentclass[twoside,leqno]{article}

\usepackage[letterpaper]{geometry}

\usepackage{ltexpprt}
\usepackage{hyperref}

\usepackage{mathtools}
\usepackage{verbatim}
\usepackage{amssymb}
\usepackage{microtype}
\usepackage{tikz}
\usepackage{amsmath}
\usetikzlibrary{arrows}
\usepackage{comment}
\usepackage{afterpage}
\usetikzlibrary{shapes.multipart,fit}

\usepackage{enumitem}
\usepackage{xspace}
\usepackage{hyperref}
\hypersetup{
    colorlinks,
    citecolor=black,
    filecolor=black,
    linkcolor=black,
    urlcolor=black
}

\newtheorem{definition}{Definition}[section]

\author{%
Peyman Afshani\thanks{Aarhus University, Denmark.}
\and
John Iacono\thanks{Universit\'e Libre de Bruxelles and New York University. Supported by Fonds de la Recherche Scientifique-FNRS under Grant no MISU F 6001 1 and NSF grant CCF-1911245.}
\and 
Varunkumar Jayapaul\thanks{Universit\'e Libre de Bruxelles. Research partially completed while at Universit\'e Libre de Bruxelles, supported by Fonds de la Recherche Scientifique-FNRS under Grant no MISU F 6001 1.} 
\and 
Ben Karsin\thanks{NVIDIA. Research partially completed while at Universit\'e Libre de Bruxelles, supported by Fonds de la Recherche Scientifique-FNRS under Grant no MISU F 6001 1.}
\and
Nodari Sitchinava\thanks{University of Hawaii at Manoa, USA. Partially supported by NSF grants CCF-1911245.}
}
\date{}

\newcommand\relatedversion{}
\renewcommand\relatedversion{\thanks{A extended abstract of this paper will appear at APOCS 2022.}} 

\title{Locality-of-Reference Optimality of Cache-Oblivious Algorithms\relatedversion{}}

\definecolor{darkgreen}{rgb}{0,0.5,0}  
\definecolor{darkpink}{rgb}{0.9,0.3,0.5}  
\definecolor{darkbrown}{rgb}{0.4,0.25,0.18}

\newcommand{\john}[1]{\textcolor{red}{$\langle${\sc John Says:} #1$\rangle$}}
\newcommand{\vkj}[1]{\textcolor{blue}{$\langle${\sc Vkj Says:} #1$\rangle$}}
\newcommand{\ben}[1]{\textcolor{darkgreen}{$\langle${\sc Ben Says:} #1$\rangle$}}
\newcommand{\nodari}[1]{\textcolor{purple}{$\langle${\sc Nodari Says:} #1$\rangle$}}
\newcommand{\peyman}[1]{\textcolor{darkbrown}{$\langle${\sc Peyman Says:} #1$\rangle$}}
\renewcommand{\nodari}[1]{}
\renewcommand{\peyman}[1]{}
\renewcommand{\john}[1]{}
\renewcommand{\vkj}[1]{}
\renewcommand{\ben}[1]{}

\newcommand{\loc}{\ell}
\newcommand{\bl}[2]{\{#2\}_{#1}}

\newcommand{\wi}{\mathcal{W}^{\textsc{opt}}}
\newcommand{\wlru}{\mathcal{W}^{\textsc{lru}}}
\newcommand{\w}{\mathcal{W}}

\newcommand{\cost}[2]{\ensuremath{Q^{#1}_{#2}}}
\newcommand{\model}{\ensuremath{\mathbb{M}}}
\newcommand{\wc}[2]{\ensuremath{W^{#1}_{#2}}}
\newcommand{\wset}{\ensuremath{\mathcal{W}^{\model}}}

\newcommand{\cco}[1][M]{\cost{\textsc{co}}{#1,B}}
\newcommand{\clru}[1][M]{\cost{\textsc{lru}}{#1,B}}
\newcommand{\clor}[1][\loc]{\cost{\textsc{LoR}}{#1}}

\newcommand{\wcco}[1][M]{\wc{\textsc{co}}{#1,B}}
\newcommand{\wclru}[1][M]{\wc{\textsc{lru}}{#1,B}}
\newcommand{\wclor}[1][\loc]{\wc{\textsc{LoR}}{#1}}

\newcommand{\dist}[1]{\delta_t(#1)}

\newcommand{\jr}[2]{{\cost{\textsc{LoR}}{#1}}(#2)}

\newcommand{\jrh}[2]{\widehat{\cost{\textsc{LoR}}{#1}}(#2)}
 \newcommand{\jlr}[1]{\cost{\textsc{LoR}}{\loc}(#1)}
\newcommand{\jlrk}[2]{\cost{\textsc{LoR}}{_{\loc_{#2}}}(#1)}

\newcommand{\locset}{\mathcal{L}}
\newcommand{\blocset}{\mathcal{L_{B}}}

\newcommand{\mP}{\mathcal{P}}
\newcommand{\mI}{\mathcal{I}}
\newcommand{\mA}{\mathcal{A}}

\newcommand{\mM}{\ensuremath{\mathcal{M}_P(B)}}

\newcommand{\PI}[1]{\mathcal{P}^{\textsc{opt}}_{^{#1}}}
\newcommand{\PLRU}[1]{\mathcal{P}^{\textsc{lru}}_{^{#1}}}
\newcommand{\coshift}[2]{{\cost{\textsc{sco}}{#1}}(#2)}
\newcommand{\ccoshift}[2][M]{{\cost{\textsc{sco}}{#1,B}}(#2)}
\newcommand{\lrushift}[2]{{\cost{\textsc{slru}}{#1}}(#2)}
\newcommand{\co}[2]{{\cost{\textsc{co}}{#1}}(#2)}
\newcommand{\lru}[2]{{\cost{\textsc{lru}}{#1}}(#2)}
\newcommand{\tj}[2]{t_{#1}(#2)}

\newcommand{\cosmooth}{memory-smooth\xspace}

\newcommand{\coopt}{\textsc{OPT}^{\textsc{co}}_{^{M,B}}}

\newcommand{\lruopt}{\textsc{OPT}^{\textsc{lru}}_{^{M,B}}}
\newcommand{\jumpopt}{\textsc{OPT}_{j}}
\newcommand{\eopt}{E_{\textsc{OPT}}}

\newcommand{\lrp}[1]{\left( #1 \right)}

\newcommand{\bstable}{\ensuremath{B}-stable\xspace}       
\newcommand{\bstability}{\ensuremath{B}-stability\xspace} 

\newcommand{\Iw}{I_w}
\newcommand{\locinv}{\loc^{^-1}}

\usepackage{thmtools}
\usepackage{thm-restate}

\usepackage{amsmath}
\usepackage{mathtools}

\DeclarePairedDelimiter{\sqbracketsX}{[}{]}

\newcommand*{\sqbrackets}{\sqbracketsX*}

\newcommand{\para}[1]{\textbf{#1}}

\DeclareMathOperator{\expop}{\mathbb{E}}
\newcommand*{\E}{\expop\sqbrackets}

\newcommand{\N}{\mathbb{N}}
\newcommand{\R}{\mathbb{R}}

\usepackage{soul}
\usepackage{verbatim}
\usepackage{color}

\begin{document}

\maketitle

\begin{abstract}
\small\baselineskip=9pt
The program performance on modern hardware is characterized by \emph{locality
of reference}, that is, it is faster to access data that is close in address
space to data that has been accessed recently than data in a random location.
This is due to many architectural features including caches, prefetching,
virtual address translation and the physical properties of a hard disk drive;
attempting to model all the components that constitute the performance of a
modern machine is impossible, especially for general algorithm design purposes.
What if one could prove an algorithm is asymptotically optimal on all systems
that reward locality of reference, no matter how it manifests itself within
reasonable limits? We show that this is possible, and that excluding some
pathological cases, cache-oblivious algorithms that are asymptotically optimal
in the ideal-cache model are asymptotically optimal in any reasonable setting
that rewards locality of reference. This is surprising as the cache-oblivious
framework envisions a particular architectural model involving blocked memory
transfer into a multi-level hierarchy of caches of varying sizes, and was not
designed to directly model locality-of-reference correlated performance.

\end{abstract}



\includecomment{onlymain}
\excludecomment{onlyproof}
\excludecomment{onlyapp}

\section{Introduction}\label{model}
Modeling memory access time of modern computers is an important area of research that lies
at the intersection of theoretical computer science, algorithm engineering, and practical aspects of computing. 
The reality is that modern computers are extremely complicated with numerous components that try to reduce 
the access time of the elements in the memory. 
Consequently, the access time 
varies by orders of magnitude depending on whether favorable conditions are met or not.
Choices in the algorithm design can highly impact reaching those favorable conditions,
which necessitates building good theoretical models of memory structures
of modern computers. 

The main direction of existing theoretical work has been on modeling of the memory hierarchy. This has been initiated by the Disk Access Model (DAM)~\cite{DBLP:journals/cacm/AggarwalV88}, which is also known as the I/O model or the External Memory (EM) model. 
The DAM assumes a (fast) memory of size $M$ and a disk (i.e., slow memory) of infinite size.
The disk stores the input in blocks of size $B$ and can be read or written to via the input-output (I/O) operations, where
each I/O transfers one block of data at unit cost. 
The analysis then typically only considers the number of I/Os (read or
written), ignoring any other computational cost.
The justification is that since a disk is so much slower than internal
memory, minimizing the number of block transfers and ignoring all else is a good
model of runtime, however, sometimes this is not a realistic assumption. 
For example, when DAM is used to model cached memory (modeling cache misses as I/O operations, 
$M$ as the size of the cache, and $B$ as the size of the cache lines), 
the relative difference in the cost of a cache misses compared to the arithmetic operations is much smaller. 
Aside from this, the DAM model has additional limitations, e.g., 
it ignores the fact that accessing adjacent blocks on a disk is in practice
much faster than two random blocks~\cite{DBLP:books/daglib/0022093} and
it models only two levels of memory.  

\begin{figure*}[t!]
  \centering
  \includegraphics[width=\textwidth]{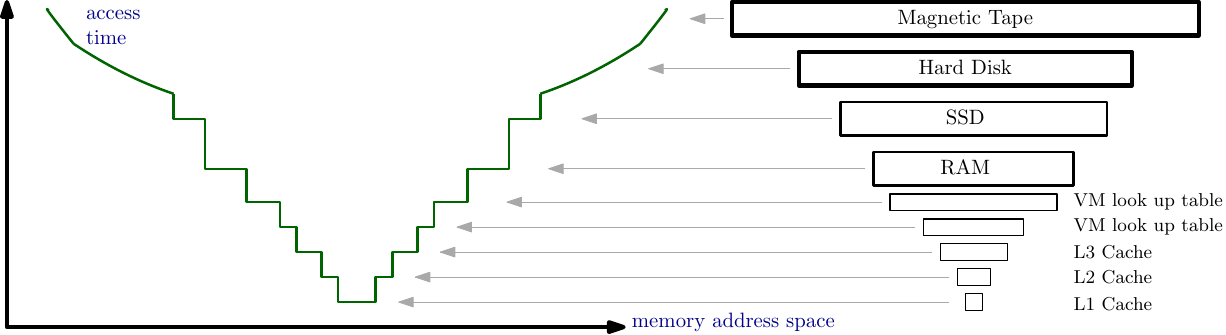}
  \caption{On the right, one can see a hierarchical model of modern computers, with multiple levels of permanent storage, look up tables for virtual memory and CPU cache. 
    From top to bottom, the levels become smaller but faster. 
    On the left, a sample access time is shown (not to scale). Here, it is assumed that each level contains a continuous portion of the elements from the
  level above it. As a result, the access time increases as we move away from the elements that are located in the fastest memory, as a complicated function of the distance.
    For caches, the access time dramatically increases when we hit an element
  that is stored in the next (slower) level.  For mechanical devices (tapes, hard disks),
 the access time increases rather smoothly
due to mechanical processes involved.}
  \label{fig:moti}
\end{figure*} 

Modeling more than two levels of the memory hierarchy is rather challenging. The big problem is that very precise 
models  (e.g., the ones defining individual parameters for each level of memory hierarchy~\cite{DBLP:journals/jcss/Valiant11}) are often too complicated, making it hopelessly difficult 
to design and analyze algorithms. Other approaches, such as the \emph{hierarchical memory model (HMM)}
\cite{DBLP:conf/stoc/AggarwalACS87,DBLP:conf/focs/AggarwalCS87},
model memory with variable access costs by assuming that the cost to access a
memory address $x$ is a non-decreasing function, $f(x)$, of the address itself. 
However, this does not accurately represent modern caches.

The most successful attempt at analyzing cache misses in multi-level cache
hierarchies is probably the {\em cache-oblivious}
framework~\cite{DBLP:journals/talg/FrigoLPR12}.  It surprisingly avoids the
complexity of modeling memory hierarchies by completely avoiding it: the
algorithms are designed ``obliviously to $M$ and $B$'', i.e., in the classical
(single-level) RAM model.  However, if such algorithms are analyzed in the
two-level DAM with the best cache-management policy, also known as the
ideal-cache model, and they happen to be efficient in terms of cache misses
with respect to the two levels of memory, then they are also efficient for
all levels of multi-level memory hierarchy. Moreover, it has been
shown~\cite{DBLP:journals/cacm/SleatorT85} that such algorithms are also
efficient for many reasonable cache management policies, e.g., least-recently
used (LRU) policy, typically implemented by hardware in practice. However, up
to now, cache-oblivious algorithms haven't been shown to optimize anything
beyond cache misses.

\para{Capturing locality of reference.} 
In a real hardware, cache utilization is only one aspect that affects program runtime. For example, Jurkiewicz and Mehlhorn~\cite{doi:10.1137/1.9781611972931.13} show that the time it takes to perform address translations for virtual memory noticeably affects the runtime of programs on real hardware. 
Figure~\ref{fig:moti} illustrates the complicated ways various hardware features affect the memory access time. 
It shows a memory hierarchy, where cheaper and larger but slower memories are placed at the top.
The functions that describe the access times of the components may have different behaviors, e.g., for a magnetic tape the access time is 
basically a linear function of the physical distance, whereas for hard disks it is a more complicated function. 

Locality of reference is a fundamental principle of computing that heavily impacts both hardware and algorithm design~\cite{DBLP:books/daglib/0022093}.  
In his widely-cited article~\cite{PeterDenning}, Peter Denning provides an overview of the history
of the concept and how it became very popular in almost all aspects of computing.
To quote him,
``The locality was adopted as an idea almost immediately by operating systems,
database, and hardware architects.''

The DAM algorithm and the cache-oblivious algorithms try to capture 
\emph{spatial locality} and \emph{temporal
locality} -- the two fundamental components of the
notion of ``locality of reference'' -- from an algorithmic point of view. 
However, given the complexity of modern hardware, the main question is if other, possibly more complex hardware features, can be modeled simply 
enough to facilitate design and analysis of algorithms. 
The approach of the HMM model~\cite{DBLP:conf/stoc/AggarwalACS87,DBLP:conf/focs/AggarwalCS87} to model the cost of access as a function of memory address is one way to keep the modeling complexity at bay. 
For instance the authors~\cite{DBLP:conf/stoc/AggarwalACS87,DBLP:conf/focs/AggarwalCS87} show that if the cost of accessing address $x$ is $\log(x)$ then  
sorting $n$ elements can be done in $O(n \log n \log\log n)$ time.
But is $\log(x)$ the correct function for modern (and ever-changing) hardware?


\subsection{Our Results.}
We propose to pick up where the previous attempts have left off by following a holistic approach.
We present the \emph{locality of reference
(LoR)} model, a computational model that looks at memory in a new way: the cost
of a memory access is based on the proximity from prior accesses via what we
call a \emph{locality function}.  

More specifically, we consider the machine as having an infinite memory with a linear \textit{address space}, i.e., memory cells are
numbered with the set of natural numbers. 
Let $E = \{e_1, \ldots, e_{|E|}\}$ be the sequence of memory addresses accessed by an algorithm $A$ while running on a given
input. 
A simple locality function $\loc$ can define the cost of accessing address $e_i$ as a function of the {\em distance} from the address $e_{i-1}$ of the preceding access, for example, $\log(|e_i - e_{i-1}|)$, $\sqrt{|e_i - e_{i-1}|}$, or any other arbitrarily function of $|e_i - e_{i-1}|$.

A specific locality function can capture the (complicated) cost of accessing the data on the hardware running the algorithm. 
For example, in the classical random access memory (RAM) model, the sequence $E$ simply takes $O(|E|)$ time, so setting $\loc$ to a constant function captures the RAM model. As we show later, by setting $\loc$ to a logarithmic function, we can model the cost of the TLB in virtual memory translation.

The goal of this paper is not to define locality functions for various models of computation. Instead,  our results show that cache-oblivious algorithms go beyond minimizing the number of cache misses. In particular, we show that 
optimal cache-oblivious algorithms are \emph{locality-of-reference optimal}, meaning,
they are asymptotically optimal with respect to {\em any} choice of locality function, subject to some mild constraints (needed to ensure that these functions reward locality of reference). In Section~\ref{querytype} we present our result in a simplified setting, in which we focus on the algorithms that do not benefit from large cache sizes. In Section~\ref{with-memory} we generalize the result to more general  algorithms that do utilize the full cache.

\subsection{Example application: optimality of van Emde Boas layout.}
We now demonstrate an example application of our results which will appear as Theorem~\ref{t:main}: namely, that the van Emde Boas layout---a layout for implicit static search trees that is optimal for cache-oblivious searching~\cite{DBLP:journals/talg/FrigoLPR12}---is also optimal for address translation on modern virtual memory architectures.
This cost of address translation is non-negligible and has been observed to 
impact the performance of fundamental algorithms such
as sorting and permuting in practice~\cite{doi:10.1137/1.9781611972931.13}.

Let us review modern virtual memory design.
Consider a machine that uses $U$ bits for addressing memory.
Virtual memory is implemented as a trie of degree $2^b$, for some parameter $b$,
where the translation process translates $b$ bits at a time, 
starting from the most significant bits of the address.
In the worst-case, one translation necessitates $U/b$ lookups, but this cost is often much lower when
TLB caching is used. 
In particular, if two addresses $e_i$ and $e_{i-1}$ share $kb$ most significant bits, then after translating the first memory address, the
first $k$ steps of the second translation  are cached in the TLB.
As a result, 
the cost function associated with virtual memory translation is essentially $\ell(|e_{i}-e_{i-1}|) = \log_b(|e_i-e_{i-1}|)$. 
This function is clearly non-negative, non-decreasing and concave, thus, satisfying the requirements of Theorem~\ref{t:main}.

The van Emde Boas layout of an implicit complete binary search tree (for brevity we'll call it {\em vEB tree}) is defined as follows. Given a complete binary search tree on $n$ vertices of height $h = \lceil \log n\rceil$, let $T_0$ be the top subtree of height $\lceil h/2 \rceil$ and let $T_1, \dots, T_{\sqrt{n}}$ be the subtrees rooted at the children of the leaves of $T_0$. Then, vEB tree is defined by placing the subtree $T_0$ in a contiguous portion of an array, immediately followed by $T_1, T_2, \dots, T_{\sqrt{n}}$, with each subtree $T_i$ laid out recursively. Search on vEB tree is known to incur $O(\log_B n)$ cache misses, where $B$ is the block size of the DAM model~\cite{DBLP:journals/talg/FrigoLPR12}.

Consider a machine that is equipped with virtual memory with parameters $U=O(\log n)$ 
and $b=O(1)$. The root-to-leaf traversal of vEB tree recursively traverses $T_0$, jumps by at most $n - \sqrt{n}$ addresses from a leaf of $T_0$ to the root of the appropriate subtree $T_i$
, and recursively traverses $T_i$.  Thus, the cost $\clor(n)$ of the traversal using the above locality of reference function $\loc$ for virtual memory can be computed using the following recurrence:
\[\clor(n) = \begin{cases} 2\cdot \clor(\sqrt{n}) + O(\log n) & \text{ if } n \ge 2 \\
                           O(1) & \text{ otherwise } \end{cases} 
\]
 This solves to $\clor(n) = O(\log n \log\log n)$.
And Theorem~\ref{t:main} implies that this is an asymptotically optimal LoR cost for virtual address translation for the above function $\loc$.

\subsection{Related work.}
The closest work to our presentation here is the Hierarchical Memory Model \cite{DBLP:conf/stoc/AggarwalACS87}. In this model, accessing memory location $x$ takes time $f(x)$. This was extended to a blocked version where accessing $k$ memory locations takes times $f(x)+k$ \cite{DBLP:conf/focs/AggarwalCS87}. In particular the case where $f(x)=\log x$ was studied and optimality obtained for a number of problems. This model, through its use of the memory cost function $f$, bears a number of similarities to ours, and it is meant to represent a multi-level cache where the user manually controls the movement of data from slow to fast memory. However, while it is able to capture temporal coherence well, even in the blocked version it does not capture fully the idea of spatio-temporal locality of reference, where an access is fast because it was close to something accessed recently.

Another model that proposed analyzing algorithm performance on a multi-level memory hierarchy is the Uniform Memory Hierarchy model (UMH)~\cite{DBLP:conf/focs/AlpernCF90}.  The UMH model is a multi-level variation of the DAM that simplifies analysis by assuming that the block size increases by a fixed factor at each level and that there is a uniform ratio between block and memory size at each level of the hierarchy. Unfortunately, this assumption is quite restrictive and doesn't hold in practice on modern hardware.   

\begin{onlymain}
\section{Preliminaries}\label{prelim}
\subsection{Models of computation}\label{uni-optimal}
Let $P$ be a problem, $\mI^P=\{I_1,I_2,\ldots\}$ be a set of valid instances (input sequences) for which problem $P$ can be solved, and $\mI_n^P  = \{I \in P: |I| = n\}$ be the subset of instances of $P$ with input of size $n$. 
Let $E(A,I) = \{e_1, e_2, \ldots\}$ be the sequence of accesses (reads and writes) to memory locations that arises by executing algorithm $A$ on instance $I$, and let $\mA_P = \{A_1, A_2, \ldots\}$ denote the set of all algorithms that correctly solve $P$, i.e., generate a correct output for every instance in $\mI^P$.  
\peyman{I suggest we add this because it is important and I think we need it. The definition of
  optimality, takes care of comparing the optimal algorithm with algorithms that have a hard-coded
  values of $M$ and $B$:\\
  We only consider algorithms that are ``oblivious'' to the model, i.e., the sequence $E(A,I)$ only depends on the 
instant $I$.}
\nodari{This is exactly what cache-oblivious algorithms are. That's why every main theorem mentions that the algorithm is cache-oblivious. But I do agree that we must state explicitly that this is the definition of cache-oblivious algorithms (i.e. define cache-obliviousness).}

\nodari{Talk about equivalence of ideal-cache and DAM models, the LRU model and the LoR model}

The general ideal-cache and LRU-cache models incorporate memory size $M$ and block size $B$ when computing the cost of an execution sequence.  The $\frac{M}{B}$ blocks stored in internal memory make up the \emph{working set}, and we define $\wset_{M,B}(E,i)$ to be the working set after the $i$-th access to the execution sequence $E$ in model $\model$.  
 In the ideal-cache model (and the DAM model) the evictions from the working set $\wi_{M,B}$ are selected such that the total cost of executing $E$ is minimized~\cite{DBLP:journals/talg/FrigoLPR12}, while in the LRU-cache model the evictions from the working set $\wlru_{M,B}$ follow the {\em least recently used} policy~\cite{DBLP:books/daglib/0022093}. 
Let $\cost{\model}{x}(E(A,I))$ denote the cost of executing algorithm $A$ on instance $I$ in model $\model$ with model parameters $x$.  Then the cost of algorithm $A$ on instance $I$ in the ideal-cache or LRU-cache models with cache parameters $M$ and $B$ is 
$$\cost{\model}{M,B}(E(A,I))=\sum\limits_{i=1}^{|E(A,I)|}~\begin{cases} 
0 \text{ if } e_i \in \wset_{M,B}(E(A,I),i-1) \\
1 \text{ otherwise}
 \end{cases}$$ for  $\model \in \{\textsc{co}, \textsc{lru}\}.$
If $E(A,I)$ only depends on the input instance $I$ and not on any of the parameters of the model (e.g.,  $M$ or $B$) then the algorithm is {\em cache-oblivious}.  
Effectively, a cache-oblivious algorithm is one that runs in the classical RAM (with one level of memory). 
On the other hand, if the algorithm explicitly uses multiple levels of memory, or if the access pattern $E(A,I)$ depends on the particular values of $M$ or $B$, then
it is called {\em cache-aware}.
A more rigorous and formal definition of the working set, cache replacement policies, and the cache-oblivious and LRU costs is included in the full version of the paper.
\end{onlymain}
\begin{onlyapp}
\section{Formal definitions of cache-oblivious and LRU cost}\label{formal-CO-LRU}
Analysis of cache-oblivious algorithms assumes $M$ to be the size of internal memory, with $\frac{M}{B}$ blocks being stored in internal memory at a given time, which we call
the \emph{working set}.  The working set is made up of blocks of contiguous memory, each containing $B$ elements.  For a given block size, $B$, we enumerate the blocks of memory by defining the block containing element $e$ as $\bl{B}{e}$ (the $\left \lfloor \frac{e}{B}\right \rfloor$-th block).  Formally, we define the working set \emph{after} the $i$-th access of execution sequence $E$ on a system with memory size $M$, block size $B$, and cache replacement policy $\mP$ (formally defined below) as $\w^{\mP}_{M,B}(E,i)$.  For simplicity of notation, we refer to the working set after the $i$-th access simply as $\w_{i}$ when the other parameters ($M$, $B$, $\mP$, and $E$) are unambiguous.

When we access an element $e_i$, if the block containing $e_i$ is in the working set (i.e., $\bl{B}{e_i} \in \w_{i-1}$), it is a \emph{cache hit} and, in the cache-oblivious model, it has a cost of 0.  However,
if $\bl{B}{e_i}$ is not in the working set, it is a \emph{cache miss}, resulting in a cost of 1.  On a cache miss, the accessed block, $\bl{B}{e_i}$ is loaded into memory, replacing an existing block, which is determined by \emph{cache replacement policy}.  We define a general cache replacement policy as a function that selects the block of the working set to evict when a cache miss occurs, i.e., for memory size $M$ and block size $B$:
\begin{align*}
\mP_{M,B}(E,\w,i) &= \begin{cases}
\bl{B}{e_k} &\text{ if }|\w| = \frac{M}{B} \\& \text{ and } \bl{B}{e_i} \not\in \w \\
\emptyset &\text{ otherwise}
\end{cases}\\
\intertext{where $\w$ is the working set, $e_i$ and $e_k$ are the $i$-th and $k$-th accesses in sequence $E$, respectively, $k<i$, and $\bl{B}{e_k} \in \w$.
For a given cache replacement policy and execution sequence, $E$, we define the working set after access $i \in E$ as}
\w_{i} &= \left( \w_{i-1} \backslash \mP_{M,B}(E,\w_{i-1},i)\right ) \cup \bl{B}{e_i}
\end{align*}
where $\mP_{M,B}(E,\w_{i-1},i)$ defines the block to be evicted and $\bl{B}{e_i}$ is the new block being added to the working set. Since a cache miss results in a cost of 1 and a cache hit has cost 0, the total cost of execution sequence $E$ is simply:
\begin{align*}
C_{M,B,\mP}(E) &= \sum_{i=1}^{|E|} \begin{cases} 
  0 & \text{if } \bl{B}{e_i} \in \w^{\mP}_{M,B}(E,i{-}1) \\
1 & \text{otherwise}\\
 \end{cases} \\
\end{align*}
For this work, we focus on the following cache replacement policies:

\begin{itemize}
\item $\PI{M,B}(E,\w,i)$: The ideal cache replacement policy with internal memory size $M$ and block size $B$.  The number of evictions (and cache misses) over execution sequence $E$ is minimized.  This is equivalent to Belady's algorithm~\cite{Belady1966} that evicts the block $\bl{B}{e_k}$ that is accessed the farthest in the future among all blocks in $\w$.
\item $\PLRU{M,B}(E,\w,i)$: The least recently used (LRU) cache replacement policy with internal memory size $M$ and block size $B$.  The evicted block, $\bl{B}{e_k}$, is the ``least recently used'' block in $\w$.  That is, $\bl{B}{e_k}$ is selected such that no element in $\bl{B}{e_k}$ has been accessed more recently than the most recently accessed element of any other block in $\w$. 
\end{itemize}

We define $\wi_{^{M,B}}(E,i)$ and $\wlru_{^{M,B}}(E,i)$ as the working sets after the $i$-th access of sequence $E$, when using the ideal and LRU cache replacement policies, respectively.  
Thus, the cache-oblivious cost (using the ideal cache replacement policy) of performing the $i$-th access of  on a system with memory size $M$ and block size $B$ is 
$\co{M,B}{E,i}~=~\begin{cases} 
0 \text{ if } \bl{B}{e_i} \in \wi_{^{M,B}}(E,i{-}1) \\
1 \text{ otherwise}\\
 \end{cases}
 $
and the total cost for the entire execution sequence $E$ is
$\co{M,B}{E} = \sum_{i=1}^{|E|} \co{M,B}{E,i}$.  We similarly define the cost with the LRU cache replacement policy for a single access $e_i$ and a total execution sequence $E$ as $\lru{M,B}{E,i}$ and $\lru{M,B}{E}$, respectively.


\begin{theorem}\label{thm:CO-LRU}
For any execution sequence, $E$, memory size $M$, and block size $B$, the total cache misses using the LRU cache replacement policy with a memory twice the size ($2M$) is \emph{2-competitive} with number of cache misses using the ideal cache replacement policy, i.e.,
$\lru{2M,B}{E} \leq 2\cdot \co{M,B}{E}$.
\end{theorem}
\begin{proof}
It follows from the work of Sleator and Tarjan~\cite{DBLP:journals/cacm/SleatorT85}.
\end{proof}

\end{onlyapp}
\begin{onlymain}
In Sections~\ref{querytype} and \ref{with-memory} we also define the cost $\clor(E(A,I))$ of an algorithm on instance $I$ in our memoryless and general {\em Locality of Reference (LoR)} models with locality function $\ell$.


Similarly, we define $\wc{\model}{x}(P, A, n) = \max_{I \in \mI_n^P} \{\cost{\model}{x}(E(A, I)\}$ as the worst-case cost of algorithm $A$ on problem instances of size $n$ for problem $P$ on model $\model$ with parameters $x$.

Analysis of both cache-oblivious and cache-aware algorithms relies on additional constraints defining the relationship between the cache parameters $M$ and $B$. For example, the analysis of the cache-aware sorting algorithms~\cite{DBLP:journals/cacm/AggarwalV88} assumes $M \ge 2B$, while the analysis of the cache-oblivious sorting algorithms~\cite{DBLP:journals/talg/FrigoLPR12} and  cache-oblivious sparse matrix dense vector (SpMV) multiplication~\cite{DBLP:journals/mst/BenderBFJV10} assume, respectively, $M \ge B^2$ and $M \ge B^{1+\epsilon}$ -- the so-called {\em tall-cache assumption}. Therefore, let $\mM$ denote the set of all values of $M$ (as a function of $B$) that satisfy such constraints for problem $P$.


\nodari{Old writeup:  \\
An understandable first reflex to these definitions is to say, ``just use the $\Theta$-notation''. 
However, there is no clear way to do so that unambiguously indicates the order of the various quantifiers, given more than one parameter in these definitions.
%
And since our proofs of optimality rely on the precise definition of this notation, we formally state these definitions here. 
}

To prove our results, we need to rigorously define algorithm optimality. 
Otherwise, as there are multiple parameters involved,
the order of the quantifiers would be unclear
and ambiguous. 
\peyman{Made this shorter}

\begin{definition}\label{def:coopt}
A cache-oblivious algorithm $A$ for problem $P$ is asymptotically  {\em CO-optimal} in the ideal-cache model with the cache parameters $B$ and $M \in \mM$ iff\footnote{A standard and understandable reaction to this definition and the math to come is to sneer at the chain of quantifiers and declare ``just write this using asymptotic notation!'' Unfortunately, such notation is not flexible enough to say exactly what we need to say, and thus we make the quantification explicit. }:
\begin{equation*}
\exists_{c>0,n_0>0} \forall_{n>n_0} \forall_{B \ge 1} \forall_{M \in \mM} \forall_{A' \in \mA_P} 
\left[ 
\wcco(P,A,n) \leq  c\cdot 
\wcco(P,A',n) \right]
\end{equation*}
\end{definition}

To implement the cache replacement policy of the ideal-cache model requires the knowledge of what an algorithm will do in the future. Instead, modern hardware caches implement an approximation of the LRU-cache model, each time evicting the block that has been accessed least recently. While it is often easier to analyze the cache misses of an algorithm in the ideal-cache model, in this work we are able to work directly in the LRU-cache model. 

\begin{definition}\label{def:lruopt}
A cache-oblivious algorithm $A$ for problem $P$ is asymptotically  {\em LRU-optimal} in the LRU-cache model with the cache parameters $B$ and $M \in \mM$ iff:
\begin{equation*}
\exists_{c>0,n_0>0} \forall_{n>n_0} \forall_{B \ge 1} \forall_{M \in \mM} \forall_{A' \in \mA_P} 
\left[ 
\wclru(P,A,n) \leq  c\cdot 
\wclru(P,A',n) \right]
\end{equation*}
\end{definition}

\peyman{Maybe add:\\
  Note that due to the order of the quantifiers above, 
  $A'$ is allowed to have the values $M$ and $B$ ``hardcoded'' 
in it, meaning, $A$ cannot be even beaten by such algorithms.}

\peyman{Why not fold both of the above definitions in one:\\
\begin{definition}\label{def:lruopt}
A cache-oblivious algorithm $A$ for problem $P$ is asymptotically optimal in cache model $\model$, with the cache parameters $B$ and $M \in \mM$ iff:
\begin{equation*}
\exists_{c>0,n_0>0} \forall_{n>n_0} \forall_{B \ge 1} \forall_{M \in \mM} \forall_{A' \in \mA_P} 
\left[ 
   \wc{\model}{M,B}(P,A,n) \leq  c\cdot 
\wc{\model}{M,B}(P,A',n) \right]
\end{equation*}
When \model{} is LRU-cache model (resp. ideal-cache model), we say $A$ is LRU-optimal (resp. CO-optimal).
\end{definition}
}
%

For completeness, however, we also prove our results in the ideal-cache model, by utilizing the following well-known resource-augmentation result of Sleator and Tarjan~\cite{DBLP:journals/cacm/SleatorT85} (which also applies to other reasonable cache-replacement policies):

\begin{lemma}\cite{DBLP:journals/cacm/SleatorT85}
Any LRU-optimal algorithm in the LRU-cache model with cache parameters $B$ and $M$ is CO-optimal in the ideal-cache model with cache parameters $B$ and $2M$. 
\end{lemma}

The equivalence in optimality between the ideal-cache and LRU-cache models relies on the cache augmentation, and says nothing about asymptotic equivalence for the same $M$. However, this is not an issue for a large class of natural problems, which can be solved using {\em \cosmooth} cache-oblivious algorithms.

\begin{definition}\label{def:cosmooth}
A cache-oblivious algorithm $A$ is {\em \cosmooth} iff increasing the memory size $M$ by a constant factor does not asymptotically change its execution cost. That is,
\begin{equation*}
    \forall_{I \in \mI^P} \forall_{c>0} \forall_{B\ge 1} \forall_{M \in \mM}
\Big[ \cco(E(A, I)) = \Theta(\cco[c\cdot M](E(A, I))) \Big],
\end{equation*}
where the $\Theta$-notation is with respect to the size of instance $I$.
\end{definition}  

Finally, we define algorithm optimality in our LoR model:
\begin{definition}\label{def:jumpopt}
Let $\locset$ be a class of functions. Algorithm $A$ for problem $P$ is asymptotically {\em LoR-optimal}  with respect to $\locset$ iff: 
\begin{equation*}
\exists_{c>0,n_0>0} \forall_{n>n_0}\forall_{\loc \in \locset}\forall_{A' \in \mA_P} 
{ \Big[ 
\wclor(P,A,n) 
\leq c\cdot 
\wclor(P,A',n) 
\Big ] }
\end{equation*}
\end{definition}

The notion of an \textsc{LoR}-optimal algorithm with respect to {\em all} possible functions would be very powerful, as such an algorithm would be asymptotically optimal on any computing device that rewards locality of reference. In this paper we come very close to achieving such optimality, requiring only a natural set of restrictions on the functions in $\locset$.


\subsection{\texorpdfstring{$B$}{B}-stable problems}\label{sub:bstable}
To show the equivalence between LoR-optimal and CO-optimal algorithms, we must avoid pathological problems with worst-case behavior that varies dramatically with different instances of the problem for different block sizes.  

We say that a problem is \emph{\bstable} if, for any algorithm $A$ that solves $P$, there is some ``worst-case'' instance $\Iw \in \mI_n^P$ that, for every $B$, has CO cost asymptotically no less than the optimal worst-case cost for that $B$, over all instances.  Formally,

\begin{definition}\label{def:bstable}
Problem $P$ is \emph{\bstable} if, for any algorithm $A \in \mA_P$ that solves $P$:
\begin{equation*}
\exists_{c>0,n_0>0} \forall_{n>n_0}\exists_{\Iw \in \mI_n^P} \forall_{B\ge 1} \forall_{M \in \mM}  
\\
 \min_{A' \in \mA_P} \wcco(P,A',n) \le c \cdot \cco(E(A,\Iw)) 
\end{equation*}
\end{definition}
Intuitively, for any algorithm that solves a \bstable problem, there must be a single instance that, for all block sizes, has cost no less than the asymptotically worst-case optimal cost.  
\peyman{Commenting this out for now. We don't seem to need it. \\
Since there are $O(\log^2 n)$ different choices for $M$ and $B$,  we can observe the following
as a corollary.
\begin{lemma}\label{lem:stable}
    Consider a problem $P$, an input size $n\ge n_0$,  
    and an algorithm $A\in \mA_P$.
    Let $B$ and $M \in \mM$ be two fixed values with 
    $W = \min_{A' \in \mA_P} (\wcco(P,A',n))$.
    If $\Pr[W \le c \cdot \cco(E(A,\Iw))]\ge 1- \frac1{2\log^2 n}$, 
    for a uniform random instance $\Iw \in \mI_n^P$,
    then $P$ is \bstable.
\end{lemma}
}

The following lemma, implied by the definition of CO-optimality (Definition~\ref{def:coopt}), states that every algorithm must have an instance on which it performs no better, asymptotically, than the CO-optimal algorithm, for every $B$:

\begin{lemma}
If an asymptotically CO-optimal algorithm $A$ solves a \bstable problem $P$, then
\begin{equation*}
\exists_{c>0,n_0>0} \forall_{n>n_0}\forall_{A' \in \mA_P}\exists_{\Iw \in \mI_n^P} \forall_{B\ge 1} \forall_{M \in \mM} \\
\Big [ \wcco(P,A,n) \le c \cdot \cco(E(A',\Iw))  \Big ]
\end{equation*}
\end{lemma}
In the full version of the paper 
we 
prove the following lemma, which shows existence of non-\bstable problems, for which our main result (Theorem~\ref{t:main}) does not hold. This justifies our classification and exclusion of these pathological cases.

\begin{lemma}\label{lemma:needbstable}
There exists a problem $P$ which is not \bstable and which has a CO-optimal algorithm which is not LoR-optimal. 
\end{lemma}

\nodari{I moved the lemma on identifying the B-stable problems to the appendix (in file why-b-stable.tex).
If we decide to remove it, remember to update the appendix.}
\subsection{\texorpdfstring{$B$}{B}-smoothed analysis}

To prove our results, we will work with the \emph{smoothed} version of the analysis. Given the access sequence $E = E(A,I)$ and an integer $s$ chosen uniformly at random in the range $[0,B)$, let $E_{B\text{-smooth}}$ be the sequence derived from $E$ where every element of $E$ is increased by $s$. Then, the {\em $B$-smoothed} costs are defined as the expected cost on the $B$-smooth sequence in the respective cache models, i.e.,  $\ccoshift{E} = \E{\cco(E_{B\text{-smooth}}}$ and $\lrushift{M,B}{E} = \E{\lru{M,B}{E_{B\text{-smooth}}}}$. 

\begin{lemma}\label{co-sco}
For any execution sequence $E$ 
$$\ccoshift{E} \le 2 \cdot \cco(E) \le 4\cdot \ccoshift{E},$$  and 
$$\lrushift{M,B}{E} \le 2\cdot \lru{M,B}{E} \le 4\cdot \lrushift{M,B}{E}.$$
\end{lemma} 
\begin{proof}
Shifting the execution sequence may cause accesses that were in the same block to become in two neighboring blocks, and accesses that were in two neighboring blocks to become in the same block. Thus, the cost may grow or shrink by a factor of two, but not more.
\end{proof}

\end{onlymain}

\begin{onlymain}

\section{Memoryless algorithms}\label{querytype}

We begin with the simplest case, namely, the {\em memoryless cache model (MCM)} where
the internal memory is just a single block, i.e. $\mM = \{ B \}$.\footnote{While this might seem
overly restrictive, a number of query algorithms on
cache-oblivious data structures are CO-optimal, because one is always
traversing further into the structure, and never comes back to revisit parts of
memory near where it has already gone.}
Note that in this case, there is no need to differentiate between LRU-cache model
and ideal-cache model because 
the working sets in both cache models, after accessing $e_i \in E(A,I)$, consist of a single block that contains $e_i$.
Thus,
the costs of the algorithm $A$ on instance $I$ in the MCM
becomes 
\begin{align*}
    \cco[B](E(A,I)) &= \clru[B](E(A,I))\\ &= \sum_{i=1}^{|E(A,I)|} \begin{cases}
 1 &\text{if } \lfloor \frac{e_{i}}{B}\rfloor \neq \lfloor \frac{e_{i-1}}{B} \rfloor \\
 0 &\text{otherwise.}
 \end{cases}
\end{align*}

This cost rewards spatial locality. Hence, in the LoR model it is natural to define the locality function $\loc$ to measure the cost of executing the sequence $E(A,I)$ as a function of the \emph{spatial distance} $|e_i- e_{i-1}|$ between accesses:
\[\jr{\loc}{E(A,I)}= \sum_{i=1}^{|E(A,I)|}\loc(|e_i-e_{i-1}|)\]

Let $\locset$ denote the set of all non-negative, non-decreasing, concave functions $\loc: \N \rightarrow \R$. 
Even though $\locset$ encompasses a wide range of (arbitrarily complicated) functions, we will show that any cache-oblivious algorithm is LoR-optimal with respect to $\locset$ if and only if it is CO-optimal in the MCM.

\subsection{Main result.}

Let $\loc_B(d) = \min\left(1,\frac{d}{B}\right)$. 
We begin by proving our result for this {\em specific} locality function $\loc_B \in \locset$ and generalize it to any $\loc \in \locset$ later. 

\begin{lemma}\label{CO-jb}
For any execution sequence $E$ and block size $B$,
$\ccoshift[B]{E} = \jr{\loc_B}{E}$.
\end{lemma}
\begin{proof}
Consider a single access $e_i \in E$. Let $d=|e_{i}-e_{i-1}|$. If $d\geq B$ then $\loc_B(d)=1$, and the $i$th term of $\cco[B](E)$ is also 1 because 
$\lfloor \frac{e_{i}+s}{B}\rfloor \neq \lfloor \frac{e_{i-1}+s}{B} \rfloor$ for all $0\leq s < B$.
If $d<B$ then $\loc_B(d)=\frac{d}{B}$, and the expected value of the $i$th term of $\cco[B](E)$ is also $\frac{d}{B}$, because 
$\lfloor \frac{e_{i}+s}{B}\rfloor \neq \lfloor \frac{e_{i-1}+s}{B} \rfloor$ for $d$ of the possible shifts $s \in [0,B)$.
 \end{proof}

\begin{corollary}\label{J_Implies_CO}
If a cache-oblivious algorithm $A$ for problem $P$ is asymptotically LoR-optimal with respect to $\locset$, then it is asymptotically CO-optimal in the MCM. 
\end{corollary}
\begin{proof}
Since $A$ is LoR-optimal, then it is within a constant factor of optimal for all locality functions, including $\loc_B(d) = \min\left(1,\frac{d}{B}\right)$. The corollary follows from Lemmas~\ref{CO-jb} and \ref{co-sco}.
\end{proof}
\end{onlymain}

\begin{onlymain}
We now show that \emph{any} locality function $\loc \in \locset$ can be represented as a linear combination of $N$ functions $\loc_B(d)$, for $N$ various values of $B$.

\begin{restatable}{lem}{lincomb}\label{lin_comb}
For every locality function $\loc \in \locset$ there
exist non-negative constants
$\alpha_1, \alpha_2 \ldots \alpha_N$
and
$\beta_1, \beta_2 \ldots \beta_N$
such that 
$
\loc(d) = \sum_{i=1}^N \alpha_i \loc_{\beta_i}(d)
$
for integers $d$ in $[1..N]$. 
\end{restatable}
\begin{proof}[Proof (Sketch)]
Let 
$\gamma_i= \begin{cases}
            2\loc(i) - \loc(i+1) & \text{ if } i = 1\\
            2\loc(i) - \loc(i+1) - \loc(i-1) & \text{ if } 2 \le i \le N-1 \\
            \loc(i) - \loc(i-1) & \text{ if } i = N
            \end{cases} $ 
and let $\alpha_i=i\gamma_i$ and 
$\beta_i=i$.   Then for any integer $x \in [1..N]$:
$
\sum_{i=1}^N \alpha_i \loc_{\beta_i}(x) = 
\sum_{i=1}^N \min\lrp{\alpha_i,\frac{\alpha_ix}{\beta_i}} 
= 
\sum_{i=1}^N \min(i\gamma_i,\gamma_i x) 
= 
\sum_{i=1}^{x-1} i\gamma_i + \gamma_x x + \sum_{i=x+1}^N \gamma_i x
$ and this expression telescopes to $l(x)$ (the full derivation can be found in the full version of the paper. 

Since function $\loc$ is non-negative and concave, all values of $\alpha_i$ and $\beta_i$ are non-negative.
\end{proof}
\end{onlymain}

\begin{onlyproof}
\subsection{Proof of Lemma~\ref{lin_comb}} \label{pf:lin_comp}
\lincomb*
\begin{proof}
Let 
\[\gamma_i= \begin{cases}
            2\loc(i) - \loc(i+1) & \text{ if } i = 1\\
            2\loc(i) - \loc(i+1) - \loc(i-1) & \text{ if } 2 \le i \le N-1 \\
            \loc(i) - \loc(i-1) & \text{ if } i = N
            \end{cases} \] 
and let $\alpha_i=i\gamma_i$ and 
$\beta_i=i$.   Since $\loc$ is a non-negative and concave, all $\gamma_i$ values are non-negative, and, consequently, all $\alpha_i$ and $\beta_i$ values are also non-negative.

Thus:
\begin{align*}
\sum_{i=1}^N \alpha_i \loc_{\beta_i}(x)
&=
\sum_{i=1}^N \min\lrp{\alpha_i,\frac{\alpha_ix}{\beta_i}}
\\
&=
\sum_{i=1}^N \min(i\gamma_i,\gamma_i x)
\\
%
&=\overbrace{\sum_{i=1}^{x-1} i\gamma_i}^{A} + \overbrace{\gamma_x x}^{B} + \overbrace{\sum_{i=x+1}^N \gamma_i x}^{C} 
\intertext{We first simplify the $A$ term, which gives us}
A &= 2\loc(1) - \loc(2) \\&\text{\quad\quad} + \sum_{i=2}^{x-1} \Big [2\loc(i)i - \loc(i-1)i - \loc(i+1)i \Big ]\\
&= \sum_{i=1}^{x-1}2\loc(i)i - \sum_{i=1}^{x-2}\loc(i)(i+1) \\& \text{\quad\quad} - \sum_{i=2}^{x}\loc(i)(i-1) \\
&= \sum_{i=2}^{x-2}\Big[ \loc(i)(2i-(i-1)-(i+1))\Big ] \\
& \text{\quad\quad}+ 2\loc(1) + 2\loc(x-1)(x-1) \\
& \text{\quad\quad}- 2\loc(1) - \loc(x-1)(x-2) \\& \text{\quad\quad}- \loc(x)(x-1)  \\
&= \loc(x-1)(2x-2-x+2) - \loc(x)(x-1) \\
&= \loc(x-1)x - \loc(x)(x-1)
\intertext{We now simplify the $C$ term from above}
C &= \sum_{i=x+1}^{N}\gamma_i x \\
&= \sum_{i=x+1}^{N-1}\Big[2\loc(i)x - \loc(i+1)x - \loc(i-1)x \Big ]  
\\& \text{\quad\quad}+ \loc(N)x-\loc(N-1)x \\
&= \left(\sum_{i=x+1}^{N-1} 2\loc(i)x - \sum_{i=x+2}^{N} \loc(i)x - \sum_{i=x}^{N-2} \loc(i)x \right) 
\\& \text{\quad\quad}+ \loc(N)x-\loc(N-1)x \\
&= \sum_{i=x+1}^{N} 2\loc(i)x - \sum_{i=x+2}^{N} \loc(i)x - \sum_{i=x}^{N} \loc(i)x \\
&= 2\loc(x+1)x\\& \text{\quad\quad} + \left(\sum_{i=x+2}^N 2\loc(i)x  - \loc(i)x - \loc(i)x \right)
\\& \text{\quad\quad}- \loc(x)x-\loc(x+1)x \\
&= \loc(x+1)x - \loc(x)x
\intertext{Combining the simplified terms, we get}
\sum_{i=1}^N \alpha_i \loc_{\beta_i}(x) &= \overbrace{\loc(x-1)x - \loc(x)(x-1)}^{A} + \overbrace{\gamma_x x}^{B}  \\& \text{\quad\quad}+\underbrace{\loc(x+1)x - \loc(x)x}_{C} \\
&= \Big(\loc(x-1)x - \loc(x)(x-1)\Big) \\& \text{\quad\quad} + \Big( 2\loc(x)x - \loc(x+1)x - \loc(x-1)x \Big) \\
&\text{\quad\quad} + \Big(\loc(x+1)x - \loc(x)x\Big) \\
&= - \loc(x)(x-1) + 2\loc(x)x - \loc(x)x  \\
&= \loc(x) 
\end{align*}
\end{proof}
\end{onlyproof}

\begin{onlymain}
\begin{restatable}{coro}{lincombsum}\label{lin_comb_sum}
For every locality function $\loc \in \locset $ there
exists a sets of $n$ non-negative constants $\alpha_{1}, \alpha_{2} \ldots \alpha_{n}$, and
$\beta_{1}, \beta_{2} \ldots \beta_{n}$ such that, for any execution sequence $E$, 
$
\sum_{i=1}^{n} \alpha_{i} \jr{\loc_{\beta_{i}}}{E} = \jr{\loc}{E}.
$
\end{restatable}
\end{onlymain}

\begin{onlymain}
\begin{restatable}{theorem}{cojump}\label{t:main}
Let $\locset$ be a set of all non-negative, non-decreasing, concave functions $l: \N \rightarrow \R$. Any cache-oblivious algorithm $A$ that solves a \bstable problem $P$ is LoR-optimal with respect to $\locset$ if and only if it is CO-optimal in the memoryless cache model.
\end{restatable}
\end{onlymain}
 
\begin{onlymain}
\begin{proof}
  The first direction follows from Corollary~\ref{J_Implies_CO}. To prove the other direction,
  consider the CO-optimal algorithm $A_\text{CO-Opt}$ and some algorithm $A$ that solves $P$.  Since $P$ is \bstable, by Definition~\ref{def:bstable},
\begin{align*}
\exists_{c,n_0}\forall_{n>n_0}\exists_{\Iw \in \mI_n^P}\forall_{B\ge1} &  \Big [
W_B(P,A_{\text{CO-Opt}},n) \leq c \cdot \cco[B](E(A,\Iw)) \Big ]
\intertext{Using the definition of the worst-case cost $W_B$, we get}
\exists_{c,n_0}\forall_{n>n_0}\exists_{\Iw \in \mI_n^P}\forall_{B\ge1} &  \Big [
\max_{I \in \mI^P_n} \cco[B](E(A_{\text{CO-Opt}},I))
\leq c \cdot 
\cco[B](E(A,\Iw)) \Big ]
\intertext{By Lemma~\ref{co-sco}, we get}
\exists_{c,n_0}\forall_{n>n_0}\exists_{\Iw \in \mI_n^P}\forall_{B\ge1} &  \Big [
\max_{I \in \mI_n^P}\ccoshift[B]{E(A_{\text{CO-Opt}},I)} \le 4c \cdot \ccoshift[B]{E(A,\Iw)} \Big ]
\intertext{and since, by Lemma~\ref{CO-jb}, for any $B$ the smoothed CO cost is equivalent to the LoR cost with the corresponding $\loc_B$ function,}
\exists_{c,n_0}\forall_{n>n_0}\exists_{\Iw \in \mI_n^P}\forall_{B\ge1} & \Big [
\max_{I \in \mI_n^P}\jr{\loc_B}{E(A_{\text{CO-Opt}},I)} \le 4c \cdot \jr{\loc_B}{E(A,\Iw)} \Big ]
\intertext{This inequality holds for all $B$ and thus all linear combinations of various $B$. For any locality function $\loc$ in the set of valid locality functions, $\locset$, consider $\alpha^\loc_1, \alpha^\loc_2, \ldots, \alpha^\loc_n$ and $\beta^\loc_1, \beta^\loc_2, \ldots, \beta^\loc_n$ given by Lemma~\ref{lin_comb}. We use the $\beta$'s as the $B$ values and the $\alpha$'s as the coefficients in the linear combination to get}
\exists_{c',n_0}\forall_{n>n_0}\exists_{\Iw \in \mI_n^P}\forall_{\loc \in \locset} 
& \text{\quad}
\sum_{k=1}^{n} \max_{I \in \mI_n^P} \Big( \alpha^{\loc}_k \jr{\loc_{\beta^{\loc}_k}}{E(A_{\text{CO-Opt}},I)}\Big)
\le c'  \cdot
\sum_{k=1}^{n} \Big( \alpha^{\loc}_k \jr{\loc_{\beta^{\loc}_k}}{E(A,\Iw)} \Big ) 
\intertext{$\Iw$ is a single instance of $P$, therefore it cannot have a greater total cost than the single instance that maximizes the cost}
\exists_{c',n_0}\forall_{n>n_0}\forall_{\loc \in \locset} 
&
\text{\quad} 
\sum_{k=1}^{n} \max_{I \in \mI_n^P} \Big (\alpha^{\loc}_k \jr{\loc_{\beta^{\loc}_k}}{E(A_{\text{CO-Opt}},I)} \Big )
\le c' \cdot
\max_{I \in \mI_n^P}\sum_{k=1}^{n}\Big( \alpha^{\loc}_k \jr{\loc_{\beta^{\loc}_k}}{E(A,I)} \Big)
\intertext{Moving the max outside the summation can only decrease the overall cost of the left side of the inequality, thus}
\exists_{c',n_0}\forall_{n>n_0}\forall_{\loc \in \locset} 
&
\text{\quad} \max_{I \in \mI_n^P} \sum_{k=1}^{n}\Big (\alpha^{\loc}_k \jr{\loc_{\beta^{\loc}_k}}{E(A_{\text{CO-Opt}},I)} \Big )
\le c' \cdot
\max_{I \in \mI_n^P}\sum_{k=1}^{n}\Big( \alpha^{\loc}_k \jr{\loc_{\beta^{\loc}_k}}{E(A,I)} \Big)
\intertext{Using Corollary~\ref{lin_comb_sum}, we get}
\exists_{c',n_0} \forall_{n>n_0}\forall_{\loc \in \locset} 
 &\text{\quad}
\max_{I \in \mI^P_n} ( \jr{\loc}{E(A_{\text{CO-Opt}},I)})
\leq c' \cdot
\max_{I \in \mI^P_n} (\jr{\loc}{E(A,I)}) 
\\
\intertext{We are considering an arbitrary algorithm $A$ that solves $P$, so this applies to all $A \in \mA_P$.  By the definition of worst-case LoR cost,}
\exists_{c',n_0} \forall_{n>n_0}\forall_{\loc \in \locset} \forall_{A \in \mA_P}  
&\text{\quad}\Big [ 
W_{\loc}(P,A_{\text{CO-Opt}},n) 
\leq c' \cdot
W_{\loc}(P,A,n) \Big ]
\end{align*}
Thus, by Definition~\ref{def:jumpopt}, $A_{\text{CO-Opt}}$ is asymptotically LoR-optimal.
\end{proof}
\end{onlymain}

\begin{onlymain}

\section{General models for algorithms with memory}\label{with-memory}
In the previous section, we considered execution sequences that did not utilize more than one block of memory, and thus locality to other than the previously accessed memory location was irrelevant. Now we generalize the model and apply it to execution sequences without this restriction.  This requires that we consider the size and contents of internal memory when computing the expected cost of an access.  






\subsection{General LoR model}\label{bivariate-jump}
To capture the concept of the working set for algorithms that use internal memory, we define bidimensional locality functions that compute LoR cost based on two dimensions: distance and time.  This bidimensional locality function, $\loc(d,\delta)$ represents the cost of a jump from a \emph{source} element, $s$, to a \emph{target} element, $t$, where $d$ and $\delta$ are the \emph{spatial distance} and \emph{temporal distance}, respectively, between $s$ and $t$.  This captures the concept of the working set by using ``time'' to determine if the source element is in memory or not. If the source is temporally close (was accessed recently) and spatially close to the target, $t$, the resulting locality cost of the jump is small.

\para{Details of the bidimensional locality functions.}
Let $\blocset$ denote the set of {\em bidimensional locality functions} that we consider. 
The functions in this set have the following properties. 
An element of this set is a function of the form $\loc(d, \delta) = \max(f(d), g(\delta))$,
where $f(d)$ is non-negative, non-decreasing, and concave, while $g(\delta)$ is a 0-1 threshold function, i.e., 
\[
g(\delta) = g_x(\delta) = \begin{cases}
0  \text{ if } 0\leq \delta \leq x \\
1  \text{ otherwise}
\end{cases}
\]
for some value $x$. 
For any $k \leq i$, the bidimensional locality cost of a jump from source
element $e_k$ to target $e_i$ in the sequence $E$ is $\loc(|e_i -
e_k|, t(E,i) - t(E,k))$, where $t(E,i)$ is the \emph{time} of the $i$-th
access.  For simplicity of notation, we define $\dist{E,k,i}$ to be the
temporal distance between the $i$-th access ($e_i$) and $k$-th access ($e_k$),
i.e., $\dist{E,k,i} = t(E,i)-t(E,k)$.  Intuitively, we can think of
$\dist{E,k,i}$ as the time from access $e_k$ to ``present'' when accessing
$e_i$.   
In addition, we require that the
functions cannot be more ``sensitive'' to temporal locality than
spatial locality, i.e., for any locality function
$\loc(d,\delta)=\max(f(d),g(\delta))$, we have that $\forall_x [ f(x) \geq g(x) ]$.
This corresponds to the \emph{tall cache assumption} $M \ge
B^2$, which is typically used in the analysis of cache-oblivious
algorithms~\cite{DBLP:conf/stoc/BrodalF03,DBLP:journals/tcs/Silvestri08}.
Therefore, we restrict the machine parameters, $M$ and $B$, to all values $B
\ge 1$ and  $\mM = \{M : M \ge B^2\}$.
A more in-depth discussion of the tall cache assumption and how it relates to the LoR model can be found in the vull version of the paper. 
\end{onlymain}
\begin{onlyapp}
\section{On the tall cache assumption}\label{tall-cache}
Cache-oblivious algorithms are analyzed for memory size $M$ and block size $B$ and the tall cache assumption simply assumes that $M \ge B^2$.  This assumption is required by many cache-obliviously optimal algorithms because many require that at least $B$ \emph{blocks} can be loaded into internal memory at a time.  It has been proven that without the tall cache assumption, one cannot achieve cache-oblivious optimality for several fundamental problems, including matrix transposition~\cite{DBLP:journals/tcs/Silvestri08} and comparison-based sorting~\cite{DBLP:conf/stoc/BrodalF03}.  Thus, we consider how this assumption is reflected in the LoR model, and whether we can gain insight into the underlying need for the tall cache assumption.

Recall that our class of bidimensional locality functions are of the form $\loc(d,\delta)=\max(f(d),g(\delta))$, where $f$ is subadditive and $g$ is a 0-1 threshold function.  In Section~\ref{sec:jump-CO} we define the locality function that corresponds to a memory system with memory size $M$ and block size $B$ to be \begin{align*}
\loc_{M,B}(d,\delta)&=\max\left (\min\left(1, \frac{d}{B}\right), \min\left(1,  \left\lfloor\frac{\delta}{M/B}\right\rfloor \right) \right)
\end{align*}
thus, for this function, $f(d)=\min\left(1, \frac{d}{B}\right)$ and $g(\delta)=\min\left(1,  \left\lfloor\frac{\delta}{M/B}\right\rfloor \right)$.  The tall cache assumption states that $M \ge B^2$, or $\frac{M}{B} > B$.   This is reflected in our locality function as the requirement that $\forall_{k\ge 0} [f(k) \ge g(k)]$.  This restriction between $f$ and $g$ implies that $\loc$ cannot be more ``sensitive'' to temporal locality than spatial locality.  That is, the LoR cost when spatial and temporal distance are equal will be computed from the spatial distance (i.e., $\loc(d,\delta) = f(d)$ if $d\ge\delta$).  Additionally, this implies that $\loc(x,x)$ is subadditive.  Intuitively, this tells us that, with the tall cache assumption, any algorithm that balances spatial and temporal locality of reference will not have performance limited by temporal locality.  Many cache-obliviously optimal algorithms aim to balance spatial and temporal locality, thus requiring the tall cache assumption to achieve optimality. 
\end{onlyapp}
\begin{onlymain}

We form our definition of time based on the amount of change that occurs to the working set.  For example, if an access causes a block of $B$ elements to be evicted, we say that time increases by 1.  Thus, time depends on the locality function, so we define the time of the $i$-th access of $E$, for the given locality function $\loc$ to be $\tj{\loc}{E,i} = \sum_{k=1}^{i-1} (\jr{\loc}{E,k})$.  That is, the time of access $e_i \in E$ is simply the sum of costs of all accesses prior to $e_i$ in sequence $E$.   We note that the time after the last access of $E$ is the total LoR cost (i.e., $\jr{\loc}{E} = \tj{\loc}{E,|E|+1}$).

Unlike the memoryless LoR cost, we cannot simply compute the cost of access $e_i$ using the distance from the previous access, $e_{i-1}$, since any of the prior accesses may be in the working set when accessing $e_i$.  Furthermore, since we no longer consider only non-decreasing execution sequences, when accessing $e_i$, there may be accesses to both the \emph{left} and \emph{right} that could be in the same block as $e_i$.   Therefore, computing $\jr{\loc}{E,i}$ using the locality function from a single source is insufficient to capture the idea of the working set, and a detailed example showing why this is the case is included in the full version of the paper. 
We define the general LoR cost of access $e_i \in E$ as $\jr{\loc}{E, i} = $
$$
 \max\left(
\left(\begin{aligned} 
&\overbrace{\min_{\substack{\forall L<i \text{ s.t.} \\ e_L \leq e_i}} \loc(e_i-e_L, \dist{E,i,L})}^{\text{left side}} 
\\
&\qquad\qquad\qquad + 
\\
&\underbrace{\min_{\substack{\forall R<i \text{ s.t.} \\ e_R \geq e_i}} \loc(e_R-e_i, \dist{E,i,R})}_{\text{right side}}
\end{aligned}\right) -1,0
\right).
$$
Intuitively, the LoR cost of access $e_i \in E$ is computed from the minimum cost jumps from both the left side and right side of $e_i$.  We note that this generalizes the LoR cost definition of the memoryless setting (Section~\ref{querytype}), as the locality function from source $e_R$ will always evaluate to 1 for non-decreasing accesses.  
\nodari{We aren't talking about ordered accesses anymore.}

This formulation has the added benefit that it lets us easily visualize an execution sequence in a graphical representation, illustrated in Figure~\ref{fig:visualize-smooth}.   
We consider a series
of accesses in execution sequence $E$ as points in a 2-dimensional plane.  The point representing access $e_i$ is plotted with the $x$ and $y$ coordinates corresponding to the spatial position, $e_i$, and the temporal position, $t(E,i)$, respectively.  
The cost of access $e_i$ is simply computed from the LoR cost with sources $e_L$ and $e_R$ (the previous access with the minimum locality function cost to the left and right, respectively).  We can visually determine which previous accesses correspond to $e_L$ and $e_R$: if a previous access is outside the gray region (i.e., $\delta > \frac{M}{B}$ or $d > B$), the cost is 1.  Otherwise, it is simply $\frac{d}{B}$.  

\begin{figure}[tbh]
\center
\includegraphics{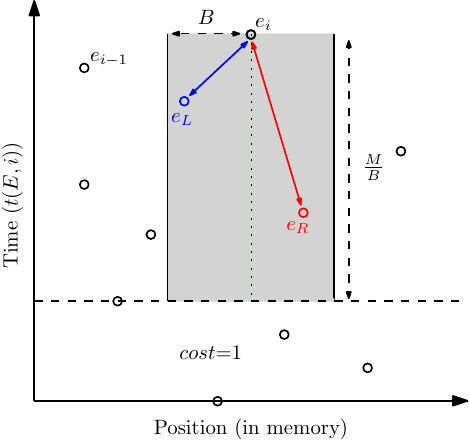}
\caption{Graphical visualization of accesses in the general LoR model with locality function $\loc_{M,B}$.  At time of access $e_i$, all prior accesses within the past time $\frac{M}{B}$ (above the dashed horizontal line) have $g(\delta) = 0$.   The locality cost to jump from any element outside the gray region has the maximum cost of 1.  In this example, there is both an $e_L$ and $e_R$ with cost $<1$.}
\label{fig:visualize-smooth}
\end{figure}


\end{onlymain}
\begin{onlyapp}
\section{A single LoR source does not represent the working set}\label{no-onefinger}

In this section, we show that computing the general LoR cost using only a single source (with the minimum cost) is insufficient to represent the working set.  Specifically, we show the potential discrepancy between such a formulation of LoR cost and the smooth LRU cost.  
Informally, we show two execution sequences where the sequence of distances from the closest previous accessed object is given by:
$$
B,\frac{B}{2},\frac{B}{4},\frac{B}{4},\overbrace{\frac{B}{8},\ldots, \frac{B}{8}}^{\text{4 times}},
\overbrace{\frac{B}{16},\ldots, \frac{B}{16}}^{\text{8 times}},\ldots, \overbrace{1,\dots, 1}^{\frac{B}{2}\text{ times}}
$$
but all the accesses in the first sequence lie within one block, while in the second, $\frac{1}{2}\log B$ blocks are accessed. This shows that the distance from the closest previous access by itself can not characterize the runtime.

We formally define this single-source definition of LoR cost of accessing $e_i$ as 
$\jrh{\loc}{E,i} = \min_{k=1}^{i-1}\loc(|e_i - e_k|, \dist{E,k,i})$.  
To show the discrepancy between this formulation and the LRU cost, we consider the specific locality function that corresponds to the LRU cost for a specific memory size $M$ and block size $B$: $\loc_{M,B}(d,\delta) = \max\left(\min\left(1,\frac{d}{B}\right), \min\left(1, \left\lfloor\frac{\delta}{M/B}\right\rfloor\right)\right)$.

Given an array of elements, $a$, located in contiguous memory, we define $a[i]$ as the $i$-th element in array $a$.  Consider an execution sequence $E$ that first accesses $a[0]$ and $a[B]$, then performs a series of \emph{stages} of accesses of elements within the range $[0,B]$.  At the first stage, $a[B/2]$ is accessed.  At the second stage, $a[B/4]$ and $a[3B/4]$ are accessed.  At stage 3, $a[B/8]$, $a[3B/8]$, $a[5B/8]$ and $a[7B/8]$ are accessed, and so on for $\log{B}$ stages.  By tall cache, we know that $M>2B$, so at any stage $k$, the blocks containing elements accessed during the previous stage are in the working set.  Thus, the LRU cost of execution sequence $E$ is $\lrushift{M,B}{E} =$
$$
\overbrace{\sum_{k=0}^{\log{B}-1}}^{\text{stages}} \overbrace{\sum_{i=0}^{2^k-1}}^{\text{accesses}} \frac{1}{B}\sum_{s=0}^{B-1}\begin{cases}
0 & \text{ if }\left \{\frac{(2i+1)\cdot B}{2^{k{+}1}}{+} s\right \}_B \\& \text{\quad is}   \left\{ \frac{2i\cdot B}{2^{k{+}1}}{+}s\right\}_B \\& \text{\quad or }  \left\{\frac{(2i+2)\cdot B}{2^{k+1}}{+}s\right \}_B \\
1 & \text{ otherwise}
\end{cases}
$$
For every access, all elements element between $a[0]$ and $a[B]$ will always be in the working set, for every shift value $s$.  Thus, the cost is 0 at each stage after the first two accesses, so $\lrushift{M,B}{E} =2$.

The single-source LoR cost, $\jrh{\loc_{M,B}}{E}$, however, depends only on the single access in the working set with the smallest spatial distance (i.e., minimum LoR cost).  Since, at each stage, the accesses from the previous stage have temporal distance $<\frac{M}{B}$, the temporal component of the locality function is always 0 and $\frac{d}{B}$ dictates the cost of each access.  At each stage, the spatial distance, $d$, decreases by a factor 2, thus
\begin{align*}
\jrh{\loc_{M,B}}{E} &= \overbrace{\sum_{k=0}^{\log{B}-1}}^{\text{stages}} \overbrace{\sum_{i=0}^{2^k-1}}^{\text{accesses}} \loc_{M,B}\left(\frac{B}{2^{k+1}}, 1\right) \\
&= \sum_{k=0}^{\log{B}-1} \left ( 2^k \cdot \loc_{M,B}\left(\frac{B}{2^{k+1}}, 1\right)\right)
\intertext{using the locality function, $\loc_{M,B}(d, \delta)$ defined above, we get}
\jrh{\loc_{M,B}}{E} &= \sum_{k=0}^{\log{B}-1} \left ( 2^k \cdot \frac{B/2^{k+1}}{B} \right ) \\
&= \sum_{k=0}^{\log{B}-1} \frac{1}{2} \\
&= \frac{\log{B}}{2}
\end{align*} 
Thus, the single-source cost formulation does not generalize the LRU cost, while using two sources does (as we prove in Lemma~\ref{equiv-LR}).
\end{onlyapp}
\begin{onlymain}

\subsection{Equivalence to cache-oblivious cost}\label{sec:jump-CO}
As with the memoryless LoR model, we first prove our result for a specific bidimensional locality function, $\loc_{M,B}(d,\delta)~=~\max\left (\min\left(1, \frac{d}{B}\right), \min\left(1,  \left\lfloor\frac{\delta}{M/B}\right\rfloor \right) \right)$, and later generalize our result to all bidimensional locality function in $\blocset$.  


\begin{restatable}{thrm}{equivLR}\label{equiv-LR}
For any cache-oblivious algorithm $A$, $B \ge 1$, $M \ge B^2$, execution sequence $E$ generated by $A$, and $\loc_{M,B}(d,\delta)~=~\max\left (\min\left(1, \frac{d}{B}\right), \min\left(1,  \left\lfloor\frac{\delta}{M/B}\right\rfloor \right) \right)$: 
$\jlrk{E}{M,B} = \lrushift{M,B}{E}$ and, consequently,
$\ccoshift{E} \leq \jlrk{E}{M,B} \leq 2\cdot \ccoshift[\frac{M}{2}]{E}$. \nodari{If there is space, move the second equation into a separate corollary.}
\end{restatable}



\begin{proof}
To prove that $\jlrk{E}{M,B} = \lrushift{M,B}{E}$, 
we consider the cost of performing access $e_i \in E$.  Assume that, when accessing $e_i$, $e_L$ is the nearest element to the \emph{left} of $e_i$ ($e_L \leq e_i$) that is in the working set, i.e., $L<i$, $e_L \in \wlru_{^{M,B}}(E,i-1)$, and $e_i-e_L$ is minimized.  If there is no such element to the left of $e_i$ in the working set, then we say that $e_L = -\infty$.   Similarly, assume that $e_R$ is the nearest element in $\wlru_{^{M,B}}(E,i-1)$ to the \emph{right} of $e_i$ ($e_i \leq e_R$), and if there is no such element, then $e_R= +\infty$.  By this definition, for any access $e_i$, we can simply consider the spatial components, because, if no element is within temporal distance $\frac{M}{B}$, the spatial distance is $\infty$ and the $\loc_{M,B}$ cost is 1.

We consider three 
possible cases for the spatial distance of $e_L$ and $e_R$ from access $e_i$: \\

\noindent
\textbf{Case 1: }$(e_i-e_L) \ge B$ AND $(e_R-e_i)\ge B$

There is no element in the working set within distance $B$ of $e_i$, then, for all alignment shifts, $0 \leq s < B$, we know that $e_i \not\in \wlru_{M,B}(E,i-1)$.  Thus,
\begin{align*}
\lrushift{M,B}{E,i} &= \E{\co{M,B}{E_{B-\text{smooth}},i}} \\
&= \frac{1}{B}\sum_{s=0}^{B-1} 1 \\
&= 1
\intertext{and the LoR model cost is}
\jlrk{E,i}{M,B} &= \max( \loc_{M,B}(e_i-e_L, \dist{E,i,L}) \\&\text{\quad}+ \loc_{M,B}(e_R - e_i, \dist{E,i,R}) - 1, 0 ) \\
&= \max( 1 + 1 - 1, 0) = 1
\end{align*} 
We note that this includes the cases where $e_L$ and/or $e_R$ do not exist, since, we set $e_L = -\infty$ and/or $e_R=\infty$, respectively, in such cases. \\


\noindent
\textbf{Case 2: }$(e_i - e_L) < B$ OR $(e_R-e_i) < B$

Only one side (left or right) is within distance $B$ of $e_i$.  W.l.o.g, assume that $(e_R - e_i) < B$ and $(e_i - e_L) \ge B$.  Since $(e_i - e_L) \ge B$, for all shifts $0\leq s < B-1$, we know that $e_L \not\in \wlru_{M,B}(E,i-1)$.  Thus, the smoothed LRU cost is simply
\begin{align*}
\lrushift{M,B}{E,i}  &= \frac{1}{B}\sum_{s=0}^{B-1} \begin{cases}
1 & \text{ if } \lfloor\frac{e_R+s}{B}\rfloor \neq \lfloor\frac{e_i+s}{B}\rfloor \\
0 & \text{ otherwise}
\end{cases}\\
&= \frac{e_R - e_i}{B}
\intertext{The LoR cost is}
\jlrk{E,i}{M,B} &= \max\left( \min\left(1, \frac{e_i - e_L}{B}\right) \right. \\&\text{\quad} + \left.\min\left(1,\frac{e_R - e_i}{B}\right) - 1, 0 \right ) \\
&= \max\left( 1 + \frac{(e_R - e_i)}{B} - 1, 0\right) \\& = \frac{(e_R - e_i)}{B}
\end{align*} 
A symmetric argument holds in the case where $(e_i - e_L) < B$ and $(e_R-e_i) \geq B$. \\

\noindent
\textbf{Case 3: }$(e_i - e_L) < B$ AND $(e_R-e_i) < B$ 


Both $e_L$ and $e_R$ are within distance $B$ of $e_i$, so the smoothed LRU cost depends on the number of alignment shifts, $s$, for which $e_i$ is not in the same block as either $e_L$ or $e_R$, i.e.,
\begin{align*}
\lrushift{M,B}{E,i} &= \frac{1}{B}\sum_{s=0}^{B-1} \begin{cases}
1 & \text{ if } \lfloor\frac{e_i+s}{B}\rfloor \neq \lfloor\frac{e_R+s}{B}\rfloor \\ & \text{ {and} } \lfloor\frac{e_i+s}{B}\rfloor \neq \lfloor\frac{e_L+s}{B}\rfloor \\
0 & \text { otherwise}
\end{cases}
\intertext{For simplicity, assume that at alignment shift $s=0$, $e_i$ is in the last location of the block of size $B$.  Thus, the shifts from $s=0$ to $s=(B-1)$ define a $2B$ range around $e_i$ (i.e., $[e_i-B, e_i+B]$).  We define $p(e_L)$ and $p(e_R)$ to be the indexes of $e_L$ and $e_R$ in this $2B$ range, respectively. 
For all $0\leq s \leq p(e_L)$, $e_i$ is in the same block as $e_L$.  Similarly, for all $(p(e_R)-B) \leq s < B$, $e_i$ is in the same block as $e_R$  Thus, the cost is simply the number of shifts, $s$, where the entire block of size $B$ containing $e_i$ is strictly between $p(e_L)$ and $p(e_R)$, i.e.,}
\lrushift{M,B}{E,i} &= \frac{1}{B} \sum_{s=p(e_L)}^{p(e_R)-B} 1 \\
&= \frac{p(e_R) - B - p(e_L)}{B}
\intertext{and, since the cost cannot be negative, this becomes}
&= \max\left(\frac{p(e_R) - p(e_L)}{B} - 1, 0\right) 
\intertext{We know that $p(e_R) = B + e_R - e_i$ and $p(e_L) = B - (e_i - e_L)$, thus}
&= \max\left( \frac{B + e_R - e_i}{B} \right. \\& \text{\qquad} - \left.\frac{B - (e_i - e_L)}{B} - 1, 0\right) \\
&= \max\left(\frac{e_R - e_i}{B} + \frac{e_i - e_L}{B} - 1, 0\right) 
\intertext{Since both $e_L$ and $e_R$ are within distance $B$ of $e_i$, this is equal to LoR cost, i.e.,}
\jlrk{E,i}{M,B} &=  \max\left( \min\left(1, \frac{e_i - e_L}{B}\right) \right. \\& \text{\qquad} +  \left. \min\left(1,\frac{e_R - e_i}{B}\right) - 1, 0 \right ) \\
&= \max\left(\frac{e_i - e_L}{B} + \frac{e_R - e_i}{B} - 1, 0\right)
\intertext{Thus, for any access access, $e_i \in E$,}
&\jlrk{E,i}{M,B} = \lrushift{M,B}{E,i}
\end{align*}
where $$\loc_{M,B}(d,\delta) =\max\left (\min\left(1, \frac{d}{B}\right), \min\left(1,  \left\lfloor\frac{\delta}{M/B}\right\rfloor \right) \right).$$
Since they are equivalent for any access, $e_i \in E$, then for any execution sequence $E$,
\begin{align*}
\jlrk{E}{M,B} &= \lrushift{M,B}{E}
\end{align*}


Since the cache-oblivious cost is computed assuming ideal cache replacement, and LRU cache replacement with twice the memory is 2-competitive with ideal cache, we have
\begin{align*}
\ccoshift{E} &\leq \lrushift{M,B}{E} \\& \leq \jlrk{E}{M,B} \\&\leq \lrushift{M,B}{E} \\& \leq 2\cdot \ccoshift[\frac{M}{2}]{E}
\end{align*}
\end{proof}
\end{onlymain}

\begin{onlymain}
\noindent
We can also prove similar asymptotic equivalence result between the LoR and the ideal-cache models for the same $M$, if we consider \cosmooth algorithms:
\begin{lemma}\label{bivariate:equal}
For any \cosmooth cache-oblivious algorithm $A$, 
$B\ge 1$, $M \ge B^2$, execution sequence $E$ generated by $A$, and $\loc_{M,B}(d,\delta)~=~\max\left (\min\left(1, \frac{d}{B}\right), \min\left(1,  \left\lfloor\frac{\delta}{M/B}\right\rfloor \right) \right)$:
$$\jlrk{E}{M,B} = \Theta\left(\ccoshift{E}\right).$$
\end{lemma}
\begin{proof}
If $A$ is \cosmooth, then $\ccoshift{E} = \Theta(\ccoshift[\frac{M}{2}]{E})$, and, by Theorem~\ref{equiv-LR}, $\jlrk{E}{M,B} = \Theta\left(\ccoshift{E}\right)$.
\end{proof}

\subsection{Main result}

We now extend our result to any bidimensional locality function $\loc \in \blocset$.
\nodari{Still need to somehow include $f(x) \ge g(x)$ in the definition of $\blocset$.}
\begin{restatable}{thrm}{lrujumpopt}\label{lru-jump-optimal}
 A cache-oblivious algorithm $A$ for a \bstable problem $P$ is LRU-optimal if and only if it is LoR-optimal with respect to $\blocset$, where $\blocset$ is a set of all functions of the form $\loc(d, \delta) = \max (f(d), g(\delta))$, where $g(\delta)$ is a 0-1 threshold function, $f(x) \ge g(x)$  for all $x \ge 0$, and $f$ is a non-negative, non-decreasing, concave function. 
\end{restatable}
\begin{proof}
If algorithm $A_{\text{LoR}}$ is LoR-optimal for all bidimensional locality functions, then it is optimal for locality functions $\loc_{M,B}$, for any $M$ and $B$.
By Theorem~\ref{equiv-LR}, it follows that $A_{\text{LoR}}$ is LRU-optimal for any $M$ and $B$.
\begin{align*}
\intertext{To prove that LRU-optimal algorithms are also LoR-optimal, we consider problem $P$ and algorithm $A_{\text{LRU}}$ that solves $P$ with optimal $LRU$ cost, i.e.,}
\exists_{c,n_0} \forall_{n>n_0} \forall_{B\geq 1} \forall_{M\geq B^2} \forall_{A \in \mA_P} 
&\text{\quad} 
W^{\textsc{LRU}}_{M,B}(P,A_{\text{LRU}},n) \le c\cdot 
W^{\textsc{LRU}}_{M,B}(P,A,n) 
\intertext{And by the definition of the worst-case cost $W$,}
\exists_{c,n_0} \forall_{n>n_0} \forall_{B\geq 1} \forall_{M\geq B^2} \forall_{A \in \mA_P} 
&\text{\quad}
\max_{I \in \mI^P_n} (\lru{M,B}{E(A_{\text{LRU}},I)})
\leq c\cdot 
\max_{I \in \mI^P_n} (\lru{M,B}{E(A,I)}) 
\intertext{Since $P$ is \bstable, there is some instance $\Iw \in \mI_n^P$ for each $A$ such that}
\exists_{c,n_0} \forall_{n>n_0} \forall_{B\geq 1} \forall_{M\geq B^2} \forall_{A \in \mA_P} 
&\text{\quad}
\max_{I \in \mI^P_n} (\lru{M,B}{E(A_{\text{LRU}},I)})
\leq c \cdot \lru{M,B}{E(A,\Iw)} 
\intertext{and by Lemma~\ref{co-sco},}
\exists_{c,n_0} \forall_{n>n_0} \forall_{B\geq 1} \forall_{M\geq B^2} \forall_{A \in \mA_P} 
&\text{\quad}
\frac{1}{2}\max_{I \in \mI^P_n} (\lrushift{M,B}{E(A_{\text{LRU}},I)})
\leq 2c \cdot \lrushift{M,B}{E(A,\Iw)} 
\intertext{therefore, by Lemma~\ref{equiv-LR}}
\exists_{c,n_0} \forall_{n>n_0} \forall_{B\geq 1} \forall_{M\geq B^2} \forall_{A \in \mA_P} 
&\text{\quad}
\max_{I \in \mI^P_n} (\jlrk{E(A_{\text{LRU}},I)}{M,B})
\leq 4c \cdot \jlrk{E(A,\Iw)}{M,B} 
\intertext{Since this inequality holds for all $\loc_{M,B}$ functions, \nodari{{\bf all}? Only of the form below, no? That's what Lemma~\ref{equiv-LR} was proven for. In fact, it should probably be part of that Lemma's definition (for some reason it's commented out in the current writeup).}  we define a series of such functions that we use to represent any bidimensional locality function.  Recall that $\loc_{M,B}$ functions are of the form }
\loc_{M,B}(d,\delta)&=\max\left (\min\left(1, \frac{d}{B}\right), \min\left(1,  \left\lfloor\frac{\delta}{M/B}\right\rfloor \right) \right)
\end{align*}
where $B\ge 1$ and $M\geq B^2$.  Consider an arbitrary bidimensional locality function $\loc(d,\delta) = \max(f(d), g(\delta))$.  By Lemma~\ref{lin_comb}, we can represent the $f(d)$ component by a linear combination of $n$ memoryless $\loc_B$ functions (and therefore using the spatial component of $\loc_{M,B}$ functions).  By our definition of bidimensional locality functions, $g(\delta) = \left\lfloor \frac{\delta}{x}\right\rfloor$, for some integer $x$.  Thus, we simply set the temporal component of every one of our $\loc_{M,B}$ functions to be $g(\delta)$.  For a given bidimensional locality function $\loc$, we define $\loc^{\loc}_{k}$ to be the $k$-th such $\loc_{M,B}$ function that we use to represent it, i.e., $\loc^{\loc}_{k} = \max(\alpha^{\loc}_k \loc_{\beta^{\loc}_k}, g(\delta))$.\footnote{We note that, because we are limited to $M \geq B^2$ for our $\loc_{M,B}$ functions, we can only construct functions where $f(k) \geq g(k)$, for all $k>0$.  However, our definition of bidimensional locality functions includes this restriction, as it corresponds to the tall cache assumption (discussed in Section~\ref{bivariate-jump}).  
}
Thus, we have
\begin{align*}
\exists_{c',n_0} \forall_{n>n_0} \forall_{A \in \mA_P} \forall_{\loc \in \blocset} 
&\text{\quad}
\sum_{k=1}^{n} \max_{I \in \mI^P_n} \Big(\jr{\loc^{\loc}_{k}}{E(A_{\text{LRU}},I)} \Big)
 \leq c'
\sum_{k=1}^{n} \Big(\jr{\loc^{\loc}_{k}}{E(A,\Iw)} \Big)
\intertext{Instance $\Iw$ cannot result in greater cost than the instance that maximizes the total cost, so}
\exists_{c',n_0} \forall_{n>n_0} \forall_{A \in \mA_P} \forall_{\loc \in \blocset} 
&\text{\quad}
\sum_{k=1}^{n} \max_{I \in \mI^P_n} \Big(\jr{\loc^{\loc}_{k}}{E(A_{\text{LRU}},I)} \Big)
\leq c'
\max_{I \in \mI^P_n} \sum_{k=1}^{n} \Big(\jr{\loc^{\loc}_{k}}{E(A,I)} \Big)
\intertext{Moving the max outside of the summation can only decrease the cost of the left hand side of the inequality, thus}
\exists_{c',n_0} \forall_{n>n_0} \forall_{A \in \mA_P} \forall_{\loc \in \blocset} 
&\text{\quad}
\max_{I \in \mI^P_n} \sum_{k=1}^{n} \Big(\jr{\loc^{\loc}_{k}}{E(A_{\text{LRU}},I)} \Big)
\leq c'
\max_{I \in \mI^P_n} \sum_{k=1}^{n} \Big(\jr{\loc^{\loc}_{k}}{E(A,I)} \Big)
\intertext{  The proof of Corollary~\ref{lin_comb_sum} applies, giving us \nodari{why is there a different constant $c''$?}}
\exists_{c'',n_0} \forall_{n>n_0} \forall_{A \in \mA_P} \forall_{\loc \in \blocset} 
&\text{\quad}
\max_{I \in \mI^P_n} (\jlr{E(A_{\text{LRU}},I)})
\leq c''\cdot 
\max_{I \in \mI^P_n} (\jlr{E(A,I)}) 
\intertext{Using our definition of the worst-case LoR cost,}
\exists_{c'',n_0} \forall_{n>n_0} \forall_{A \in \mA_P} \forall_{\loc \in \blocset} 
&\text{\quad}
W_{\loc}(P,A_{\text{LRU}},n)
\leq c'' \cdot 
W_{\loc}(P,A,n) 
\end{align*}
Therefore, any LRU-optimal algorithm is also LoR-optimal.
\end{proof}
\end{onlymain}
\begin{onlymain}


\begin{restatable}{thrm}{cosmoothoptiff}\label{cosmooth-opt-iff}
A \cosmooth cache-oblivious algorithm $A$ for a \bstable problem $P$ is CO-optimal if and only if $A$ is LoR-optimal with respect to $\blocset$, where $\blocset$ is a set of all functions of the form $\loc(d, \delta) = \max (f(d), g(\delta))$, where $g(\delta)$ is a 0-1 threshold function, $f(x) \ge g(x)$  for all $x \ge 0$, and $f$ is a non-negative, non-decreasing, concave function.
\end{restatable}
\begin{proof}
Since the cache-oblivious model assumes ideal cache replacement, for any execution sequence $E$,
$\ccoshift{E} \leq \lrushift{M,B}{E} \leq 2\cdot \ccoshift[\frac{M}{2}]{E}$.
Since algorithm $A$ is \cosmooth, for any execution sequence $E$ generated by $A$,
$\ccoshift[\frac{M}{2}]{E} = \Theta(\ccoshift{E})$.
Therefore,
$\lrushift{M,B}{E} = \Theta(\ccoshift{E})$.
Since the LRU cost and CO cost are asymptotically equivalent for every execution sequence generated by $A$, then $A$ is asymptotically LRU-optimal if and only if it is asymptotically CO-optimal and, by Theorem~\ref{lru-jump-optimal}, $A$ is LoR-optimal if and only if it is CO-optimal.
\end{proof}

\end{onlymain}

\section{Conclusion}\label{conclusion}
Despite the increasing complexity of modern hardware architectures, the goal of many design and optimization principles remain the same: maximize locality of reference.  Even many of the optimization techniques used by modern compilers, such as branch prediction or loop unrolling~\cite{DBLP:books/daglib/0022093}, can be seen as methods of increasing spatial and/or temporal locality.  
As we demonstrated in this work, cache-oblivious algorithms do just that, suggesting that the performance benefits of such algorithms extend beyond what was originally envisioned.

That is to say, although we have introduced a new way to model computation via locality functions, we are not advocating algorithm design and analysis using locality functions. Instead, though our transformations, we have shown that creating the best possible algorithms in the existing cache-oblivious models is they right way to design algorithms not just for a multi-level cache, but for any locality-of-reference-rewarding system. One can thus conclude that the cache-oblivious model is better than we thought it was.


\newpage
\bibliography{c-o,paper}

\excludecomment{onlymain}
\includecomment{onlyproof}
\excludecomment{onlyapp}


\newpage
\appendix
\section{Deferred proofs}\label{proofs}

\excludecomment{onlymain}
\excludecomment{onlyproof}
\includecomment{onlyapp}

\section{Necessity of \texorpdfstring{$B$}{B} stability} \label{s:needbstable}
The following lemma shows that  Theorem~\ref{t:main} would not hold if the restriction to $B$-stable algorithms were to be removed. 

\begin{lemma} \label{l:needbstable}
There exists a problem $P$ which is not \bstable and which has a CO-optimal algorithm which is not LoR optimal. 
\end{lemma}

\begin{proof}
Here we demonstrate a toy problem that meets the requirements of the lemma while also illustrating the unnaturalness of such problems. 
It has two candidate algorithms, one which has the same runtime on each instance, and a second one that for each instance has some values of $B$ that it runs faster than the first algorithm, and some that it runs more slowly than the first algorithm on, asymptotically. Thus for each $B$ the worst-case time of the first algorithm is better than the second, but there is no single bad instance for the second algorithm.

Consider a problem $P$ and a set $\mathcal{A}$ of two cache-oblivious algorithms $A_1$ and $A_2$. The problem, given an $n$, has a set of $n$ instances $\mathcal{I}_n=I^n_1, I^n_2, \ldots I^n_n$. The runtimes of the two algorithms are given as follows:
\begin{align*}
\cco(E(A_1,I^n_i)) &= \Theta \left( \min\left(  \frac{n \log n \log \log n} {\log i}, 
\frac{i  \cdot n \log n \log \log n}{ B \log i } \right)\right)\\ 	
\cco(E(A_2,I^n_i))  &= \Theta\lrp{\frac{n \log n \log \log \log n}{\log B}}
\end{align*}

These runtimes can be realized through an appropriately twisted problem definition that forces an algorithm for each instance to read all elements in one of two sets of memory locations in order to be considered a valid algorithm.
In particular our problem admits two algorithms, one of which, $A_2$, can solve any instance by performing $n\log n \log \log \log n$ reads in memory generated by $n \log \log \log n$ searches in a van Emde Boas search structure, and the other, $A_1$, by reading at memory locations generated by an arithmetic progression, where the step and number of locations depends on the instance.

Accessing $k$ memory locations evenly spaced $\sigma$ apart takes $\Theta(1+\min\left( k,\frac{k \sigma}{B} \right))$ I/Os in the ideal-cache model; thus the desired runtime of algorithm $A_1$ on instance $I^n_i$ can be forced by having the algorithm $A_1$ instance $I_i$ read $\frac{n \log n \log \log n}{\log i}$ memory locations evenly spaced $i$ apart.

What are the worst-case runtimes of these algorithms?

\begin{align*}
\wcco(P,A_1,n) &= \max_{I^n_i \in  \mathcal{I}_n} \cco(E(A_1,I^n_i ))
\\
&= \max_{I^n_i \in  \mathcal{I}_n}\min
\overbrace{\left(  \frac{n\log n \log \log n} {\log i},
\frac{i\cdot n \log n \log \log n}{ B \log i } \right)}^{\text{Equal when $i=B$}}
\\
&= \frac{n \log n \log \log n}{\log B}
\\
\wcco(P,A_2,n)&=\frac{n\log n \log \log \log n}{\log B}
\end{align*}
So, looking at these two algorithms, $A_2$ is clearly the worst-case optimal in the CO model. 

Now, recall the definition of \bstability (Definition~\ref{def:bstable}): Problem $P$ is \emph{\bstable} if, for any algorithm $A$ that solves $P$,
$$
\exists_{c>0,n_0>0} \forall_{n>n_0}\exists_{\Iw \in \mI_n^P} \forall_{B\ge 1} \forall_{M \in \mM}  \text{\quad}
\min_{A' \in \mA_P} \wcco(P,A',n) \le c \cdot \cco(E(A,\Iw))  
$$

Applying this to our problem gives:
\begin{align*}
\exists_{c>0,n_0>0} \forall_{n>n_0}\exists_{I^n_i \in \mI_n^P} \forall_{B\ge 1} \forall_{M \in \mM}
 \text{\quad} &  \min\left(\frac{n\log n \log \log n}{\log B}, \frac{n\log n \log \log \log n}{\log B}\right) 
\\ &\hspace{1cm} \le  
c \cdot\min\left(  \frac{n\log n \log \log n} {\log i},
\frac{i n\log n \log \log n}{ B \log i } \right)
%
%
\end{align*}

This is false for all choices of ${I^n_i \in \mI_n^P}$. Specifically, if $i\geq \log n$, then setting $B=2$ and using the first term of the min gives the following contradiction for any $c$ as $n$ grows:
\begin{align*}
\frac{n \log n \log \log \log n}{\log 2}
&\leq 
c \cdot\min\left(  \frac{n\log n {\log \log n}} {\log i},
\frac{in \log n {\log \log n}}{ B \log i } \right) \\
&\leq
c \cdot  \frac{n\log n {\log \log n}} {\log i}
\\&\leq
c \cdot  \frac{n\log n {\log \log n}} {\log \log n}
\\&=
c \cdot  {n \log n } 
\end{align*}
And, if $i\leq \log n$, then setting $B=n$ and using the second term on the right gives the following contradiction:
\begin{align*}
\frac{ n \log n \log \log \log n}{\log n}
&=n \log \log \log n\\&\leq 
c \cdot\min\left(  \frac{n\log n {\log \log n}} {\log i},
\frac{in \log n {\log \log n}}{ B \log i } \right)
\\&\leq
c \cdot\frac{n \log n \log n {\log \log n} }{ n }
\\&=c \cdot {\log^2 n \log\log n} 
\end{align*}
\nodari{we cannot claim that $i/\log i \le \log n /\log\log n$, but only that $i/\log i \le \log n$.}

This concludes the proof that $P$ is not $B$ stable. We now argue that while $A_2$ is CO-optimal for $P$, and $A_1$ is not, with locality function $\ell(d)=\log(d)$ the reverse is true as $A_1$ will have the asymptotically better runtime with this locality function in the LoR model.

In the introduction we mentioned that the LoR runtime with locality function $\ell(d)=\log d$ for searching in a vEB structure is $\Theta(\log n \log \log n)$, thus since $A_2$ does this $n \log \log \log n$ times its cost is $\Theta(n \log n \log \log n \log \log \log n)$. 
Algorithm $A_1$ is easy to analyze as on instance $I_i^n$ it accesses $\frac{n \log n \log \log n}{\log i}$ memory locations evenly spaced $i$ apart, thus its cost is
$$
\frac{n \log n \log \log n}{\log i} \cdot \log i = \log n \log \log n
.
$$
Thus the CO-optimal $A_2$ has a LoR runtime (with $\ell(d)=\log d$) of $\Theta(\log n \log \log n \log \log \log n)$ which is a $\Theta(\log \log \log)$ factor worse than the non-CO-optimal $A_1$  with LoR runtime of $\log n \log \log n$. Since $A_2$ is not optimal for one locality function, it can not be optimal for all valid locality functions.
\end{proof}

What made this problem not \bstable ? It was the fact that every instance was constructed to be faster for one algorithm for some values of $B$ and slower for others than the optimal worst-case algorithm. In this example $A_1$ ran instance $I_i^n$ slower than $A_2$ for  $B$ close to $i$ and faster than $A_2$ for $B$ far from $i$. However, it is far from natural to have an instance in effect encode faster-than-worse-case performance on selected $B$'s. In a standard data structure query, such as ``what is the predecessor of a given item in an ordered set,'' the query item itself has nothing that combined with the problem definition allows a query to encode a preference for fast execution for certain $B$'s in a non-optimal algorithm. We note that this is very different than some algorithms which may ``hard-code'' some instances and make them fast; this does not pose a problem with regards to $B$-stability as this makes this instance fast for all values of $B$.

\end{document}